\documentclass[11pt]{article}
\usepackage{times}
\usepackage{comment}

\usepackage[margin=1in]{geometry}
\usepackage{amsthm}
\usepackage{bbm}
\usepackage[linesnumbered,ruled,vlined]{algorithm2e}
\usepackage{tikz}

\usepackage{ifthen}
\usepackage{graphicx}
\usepackage{setspace}
\usepackage{caption}
\usepackage{subcaption}
\usepackage{amsmath}
\usepackage{listings}
\usepackage{indentfirst}
\usepackage{algpseudocode}

\usepackage{amssymb}
\usepackage{mathtools}
\usepackage{floatrow}
\usepackage{url}
\usepackage{hyperref}
\hypersetup{
    colorlinks=true,
    linkcolor=blue,
    filecolor=magenta,      
    urlcolor=cyan,
    citecolor=blue,
}

\newtheorem{theorem}{Theorem}
\newtheorem{corollary}[theorem]{Corollary}
\newtheorem{lemma}[theorem]{Lemma}
\newtheorem{proposition}[theorem]{Proposition}
\newtheorem{definition}[theorem]{Definition}

\newcommand{\poly}{\mathrm{poly}}

\newcommand{\rndd}[1]{\left\lfloor {#1} \right\rfloor_\ell}
\newcommand{\rndu}[1]{\left\lceil {#1} \right\rceil_\ell}

\newcommand{\cA}{\mathcal{A}}
\newcommand{\cD}{\mathcal{D}}

\DeclareMathOperator*{\bP}{{\mathbb{P}}}
\DeclareMathOperator*{\bE}{\mathbb{E}}
\DeclareMathOperator*{\E}{\mathbb{E}}
\DeclareMathOperator*{\TV}{\mathrm{TV}}
\DeclareMathOperator*{\KL}{\mathrm{KL}}
\DeclareMathOperator*{\var}{\mathrm{Var}}
\DeclareMathOperator*{\IC}{\mathrm{IC}}

\newcommand{\eps}{\epsilon}

\title{On the amortized complexity of approximate counting}

\author{Ishaq Aden-Ali\thanks{UC Berkeley. \texttt{adenali@berkeley.edu}. Supported by ONR DORECG award N00014-17-1-2127.}\and Yanjun Han\thanks{MIT. \texttt{yjhan@mit.edu}. Supported by the Berkeley-Simons Research Fellowship and Norbert Wiener postdoctoral fellowship.}\and Jelani Nelson\thanks{UC Berkeley. \texttt{minilek@berkeley.edu}. Supported by NSF award CCF-1951384, ONR grant N00014-18-1-2562, ONR DORECG award N00014-17-1-2127, and a Google Faculty Research Award.}\and Huacheng Yu\thanks{Princeton University. \texttt{yuhch123@gmail.com}.}}  

\begin{document}
\setcounter{page}{0}

\maketitle

\thispagestyle{empty}

\begin{abstract}
    Naively storing a counter up to value $n$ would require $\Omega(\log n)$ bits of memory. Nelson and Yu~\cite{NelsonY22}, following work of Morris~\cite{Morris78}, showed that if the query answers need only be $(1+\epsilon)$-approximate with probability at least $1 - \delta$, then $O(\log\log n + \log\log(1/\delta) + \log(1/\eps))$ bits suffice, and in fact this bound is tight. Morris' original motivation for studying this problem though, as well as modern applications, require not only maintaining one counter, but rather $k$ counters for $k$ large. This motivates the following question: for $k$ large, can $k$ counters be simultaneously maintained using asymptotically less memory than $k$ times the cost of an individual counter? That is to say, does this problem benefit from an improved {\it amortized} space complexity bound?
    
    We answer this question in the negative. Specifically, we prove a lower bound for nearly the full range of parameters showing that, in terms of memory usage, there is no asymptotic benefit possible via amortization when storing multiple counters.
    Our main proof utilizes a certain notion of ``information cost'' recently introduced by Braverman, Garg and Woodruff~\cite{BravermanGW20} to prove lower bounds for streaming algorithms.
\end{abstract}

\section{Introduction}
Maintaining a counter subject to increments is one of the most basic data structural problems in computer science. If the counter value stays below $n$, then the naive solution of maintaining the counter explicitly consumes $O(\log n)$ bits of memory. In 1978, Morris devised a new {\it approximate counting} algorithm which could maintain such a dynamic counter using much fewer bits in expectation \cite{Morris78}. His algorithm uses a number of random bits which itself is a random variable, whose expectation is $O(\log\log n + \log(1/\eps) + \log(1/\delta))$. At query time, this {\it Morris Counter} returns the correct counter value up to $(1+\eps)$-multiplicative error with failure probability at most $\delta$. Recently, Nelson and Yu \cite{NelsonY22} showed that a slight alteration to Morris' algorithm improves the dependence on $\delta$, consuming only $O(\log\log n + \log(1/\eps) + \log\log(1/\delta))$ bits of memory. They also proved that this bound is asymptotically optimal.

Morris' original motivation for developing his approximate counter in the mid 1970s was, for each of $26^3$ possible trigrams (sequences of three characters), to count the number of occurrences of that trigram amongst words in some text; these counts were then used by \textsf{typo}, an early Unix spellchecker (see more on this history in a survey by Lumbroso \cite{Lumbroso18}). Approximate counters still find use in modern-day applications. For example, Redis is a highly popular in-memory database which is often used as a cache, and one built-in cache replacement policy is Least Frequently Used. To implement LFU, each item stored in cache must have an associated counter, counting the number of queries to that item over some time window. These counters are currently implemented in Redis using a variant of the Morris Counter \cite{Kolevska18}.

Both the original application of Morris as well as modern-day applications of approximate counting, such as that described in Redis, need not store just one counter but many. Even in the 1970s, Morris' computer would have had enough memory to count the number of occurrences of a single trigram, but the difficulty was keeping counts of all $26^3$ counters in memory simultaneously. Similarly, modern-day computers and mobile devices clearly also contain enough memory to store a single counter for any reasonable counter value, but one may wish to minimize the space per counter in an application where many counters must be stored simultaneously, as in the case of Redis being used to cache a large number of items. Thus even the original motivation of Morris inspires the following question:

\medskip

{\it \noindent What is the \textbf{amortized} space complexity of approximate counting? Specifically, for $k$ large, is it possible to maintain $k$ approximate counters using less memory than $k$ times that of an individual approximate counter?}

\medskip

Typically such lower bounds are obtained via a direct sum argument in combination with an information complexity lower bound \cite{ChakrabartiSWY01}. Specifically, a common way to carry out proving such lower bounds in the streaming context is as follows \cite{BarYossefJKS04}: first, devise a communication game that you believe is hard (requires a lot of communication) to solve, and which reduces to the streaming problem (i.e., a low-memory streaming algorithm would imply a low-communication protocol).
Suppose for illustration that this communication game is one-way, with Alice receiving input $X$ who sends a single message $\Pi$ to Bob who holds input $Y$. Here we use distributional complexity, in which $(X,Y)$ is drawn from some distribution $\mathcal D$, and the goal is to devise a protocol in which Bob computes some $f(X,Y)$ with success probability at least $1-\delta$. Now, the so-called {\it external information complexity} of $\Pi$ is defined as $I(\Pi ; X,Y)$, which is at most $|\Pi|$ and thus it suffices to lower bound $I(\Pi; X,Y)$. Once one establishes a lower bound on this quantity, then for a sequence $\{(X_i,Y_i)\}_{i=1}^k$ of $k$ of independently drawn such inputs, we have the following inequality \cite[Theorem 1.5]{ChakrabartiSWY01}: $I(\Pi ; X_1,Y_1,X_2,Y_2,\ldots,X_k,Y_k) \ge \sum_{i=1}^k I(\Pi ; X_i, Y_i)$. Note though that if we want a lower bound for $\Pi$ solving the $k$-fold problem, then we cannot immediately invoke the single-instance lower bound to bound each summand on the right hand side (since this $\Pi$ is a protocol solving the $k$-fold problem, not a single instance!). One then typically completes the argument by showing how to efficiently embed the single-instance communication game into that of $k$ parallel instances, to then lift the single-instance lower bound to a lower bound for the $k$-fold problem as outlined above.

The trouble with obtaining our desired lower bound for approximate counting via the above paradigm is that the memory lower bound of \cite{NelsonY22} was not proven via information complexity, and thus it is not clear how to execute the above direct sum argument. Rather, the lower bound proof there followed an approach similar to the proof of the pumping lemma for regular languages \cite{Sipser97}.

\medskip

\noindent \textbf{Our Contribution:} We answer the above question in the negative for nearly the full range of parameters (for technical reasons in our analysis, we require the restriction $\epsilon < 1/\log\log(1/\delta)$). Specifically, for this range of $\epsilon$ we prove a strong memory lower bound demonstrating that there is no asymptotic benefit from amortization when maintaining multiple approximate counters. Our main approach is to first establish a novel information complexity type lower bound for a single approximate counter, which we then lift to multiple counters using a standard direct sum argument, by constructing an embedding of the single-counter problem into the $k$-fold problem.

\medskip

Below we formally define the problem we solve, then state our main result. 
The approximate counting problem for a single counter is to maintain a multiplicative approximation to a count $N$ (initially 0) that undergoes a sequence of increment operations.
Specifically, the algorithm should support two operations: \texttt{increment}() increments $N$ by $1$, and \texttt{query}() returns
a value $\hat{N}$ such that $\bP(|\hat{N} -N| \ge \eps N) \le \delta$; that is, at any point in time the data structure should be able to provide a $(1+ \eps)$ multiplicative approximation with probability at least $1-\delta$.
Our generalization to approximating $k$ counters is as follows:
\begin{definition}[$(k,\eps,\delta)$-approximate counter]
Let $N_1, \dots, N_k$ be $k$ counters all initially set to $0$.
We say a randomized algorithm $\cA$ is a $(k,\eps,\delta)$-approximate counter if it supports two operations: \texttt{increment}($i$) increments $N_i$ by $1$, and \texttt{query}($i$) returns a value $\hat N_i$ such that
\[
\bP\left(|\hat{N}_i - N_i| \ge \eps N_i \right) \le \delta.
\]
When $k =1$, we will simply call this an $(\eps, \delta)$-approximate counter.
\end{definition}

The following is then our main theorem, when here and henceforth we use $n$ in lower bounds to denote an upper bound on the number of \texttt{increment} operations:

\begin{theorem}
For any $\delta < c_1$ and $\eps < c_2/\log\log(1/\delta)$, if $k < c_3 n$ then after a sequence of at most $n$ updates any $(k,\eps,\delta)$-counter must use $\Omega(k\min\{\log(n/k), \log\log(n/k) + \log(1/\eps) + \log\log(1/\delta)\})$ bits of memory in expectation. Else if $k > c_3 n$, then a lower bound of $\Omega(n\log(2 k/(c_3 n)))$ holds; $c_1, c_2, c_3 \ge 0$ are universal constants. Furthermore, for both ranges of $k$, these lower bounds are tight up to constant factors.
\end{theorem}

\subsection{Preliminaries and notation}
We write $X \sim \cD$ to represent a random variable $X$ sampled from the probability distribution $\cD$.
If the distribution of $X$ has not been explicitly defined, we write $P_X$ to the corresponding probability distribution of $X$.
For two probability distributions $P$ and $Q$ defined on the same domain, we write $\TV(P,Q)=\int |dP-dQ|$ to be their total variation distance and $\KL(P\| Q)=\int dP\log(dP/dQ)$ their Kullback-Leibler divergence.
We frequently make use of Pinsker's inequality $\TV(P, Q) \leq \sqrt{\KL(P\|Q)/2}$.
For random variables $X$ and $Y$, we write $H(X)=\E_{x\sim P_X}[\log(1/P_X(x))]$ to denote the (Shannon) entropy and $H(X \mid Y)=\E_{(x,y)\sim P_{XY}}[\log(1/P_{X\mid Y}(x\mid y))]$ to be the conditional (Shannon) entropy.
The mutual information between two random variables $X$ and $Y$ is $I(X;Y) := H(X) - H(X \mid Y) = H(Y) - H(Y \mid X)$. 
The conditional mutual information is $I(X;Y \mid Z) := H(X \mid Z) - H(X \mid Y, Z)$.
We frequently make use of the following inequality:
\[
I(X ; Y \mid Z) = \bE_{X, Z}\left[\KL(P_{Y \mid Z} \| P_{Y \mid X, Z})\right].
\]
We also use the following well known facts about conditional mutual information \cite{CoverT06}:
\begin{proposition}[Chain rule]\label{prop:chain_rule}
For random variables $X_1, X_2, Y, Z$ we have
\[
I(X_1,X_2 ; Y \mid Z) = I(X_1 ; Y \mid Z, X_2) + I(X_2 ; Y \mid Z).
\]
\end{proposition}

\begin{proposition}[Superadditivity]\label{prop:superadd}
Let $X_1, \cdots, X_n$, $Y$ and $Z$ be random variables such that $X_1, \dots, X_n$ are conditionally independent given $Z$. 
Then
\[
\sum_{i=1}^{n}I(X_i ; Y \mid Z) \le I(X_1, \dots, X_n ; Y \mid Z).
\]
\end{proposition}

Lastly, $c,c_1,c_2,\ldots>0$ represent universal constants that may change from statement to statement.

\subsection{Proof overview}

Now we present an overview of our lower bound proof.
It turns out that the terms $\Omega(k(\log\log (n/k)+\log(1/\eps)))$ in the lower bound can be proved using a similar argument to the single counter case, which holds even in the offline setting, and the main challenge is to prove the dependence on $\delta$, $\Omega(k\log\log(1/\delta))$.
The previous proof, for single counter, of $\Omega(\log\log(1/\delta))$ is based on a pumping-lemma argument, which crucially uses the fact that all updates are exactly the same, i.e., incrementing the counter.
However, this no longer holds with multiple counters -- we may increment any one of the $k$ counters each time, and there are $k$ different updates we can perform.
The previous argument fails to generalize.

As we mentioned in the introduction, we first (re)prove an \emph{information theoretic} lower bound for single counter, then apply the \emph{superadditivity of mutual information} for independent variables to obtain the direct-sum result.
For simplicity, in this overview we will focus on the case where $\delta=2^{-\Theta(n/k)}$, and prove an $\Omega(k\log\log(1/\delta))=\Omega(k\log (n/k))$ lower bound when $\eps\leq O(1/\log (n/k))$.
This case captures all the main ideas, and generalizing to the full lower bound is straightforward.

To facilitate the argument, we first slightly reformulate the problem:\footnote{For simplicity, the overview uses slightly different notation and parameters than the actual proof.} Consider a stream with $T$ batches of updates (think $T=(n/k)^{0.1}$), and in each batch $i$, the inputs are $k$ nonnegative integers $X_{i,c}$ for $c\in[k]$, which are the number of increments we apply to each counter in this batch.
Since the batch number $i$ takes only $O(\log T)$ bits to maintain, we may assume without loss of generality that it is given to the algorithm for free.
Clearly if there is an algorithm that approximately maintains $k$ counters, this reformulation also admits a solution using the same space, provided  $\sum_{i,c}X_{i,c}\leq n$.
After the reformulation, the single counter problem simply has $T$ nonnegative numbers $X_i$ as the input stream, and we would like to approximate their sum using small space provided that the sum is at most $\poly\, T$, again assuming that $i$ is given to the algorithm for free.

We will design a hard distribution over the input streams, and analyze the failure probability and measure ``information'' with respect to this distribution.
It is crucial that the notion of information cost we use for streaming algorithms is chosen carefully.
The information cost defined in~\cite{BravermanGW20} turns out to be the right measure for our problem.
Fix a streaming algorithm, and let $M_i$ be its memory state immediately after processing $X_i$ (see Figure~\ref{fig:general-stream}).
When all numbers $X_i$ are independent, which will be the case for our distribution, the \emph{information cost} of this algorithm is defined as follows (see Figure~\ref{fig:info-cost})
\[
	\IC:=\sum_{i=1}^T\sum_{j=1}^{i}I(M_i; X_j \mid M_{j-1}).
\]

\begin{figure}[H]
    \centering
    \tikzset{every picture/.style={line width=0.75pt}} 

\begin{tikzpicture}[x=0.75pt,y=0.75pt,yscale=-1,xscale=1]

\draw   (140,153) .. controls (140,139.19) and (151.19,128) .. (165,128) .. controls (178.81,128) and (190,139.19) .. (190,153) .. controls (190,166.81) and (178.81,178) .. (165,178) .. controls (151.19,178) and (140,166.81) .. (140,153) -- cycle ;
\draw    (124.5,185.5) -- (143.97,173.11) ;
\draw [shift={(146.5,171.5)}, rotate = 147.53] [fill={rgb, 255:red, 0; green, 0; blue, 0 }  ][line width=0.08]  [draw opacity=0] (8.93,-4.29) -- (0,0) -- (8.93,4.29) -- cycle    ;
\draw   (240,153) .. controls (240,139.19) and (251.19,128) .. (265,128) .. controls (278.81,128) and (290,139.19) .. (290,153) .. controls (290,166.81) and (278.81,178) .. (265,178) .. controls (251.19,178) and (240,166.81) .. (240,153) -- cycle ;
\draw    (224.5,185.5) -- (243.97,173.11) ;
\draw [shift={(246.5,171.5)}, rotate = 147.53] [fill={rgb, 255:red, 0; green, 0; blue, 0 }  ][line width=0.08]  [draw opacity=0] (8.93,-4.29) -- (0,0) -- (8.93,4.29) -- cycle    ;
\draw   (340,153) .. controls (340,139.19) and (351.19,128) .. (365,128) .. controls (378.81,128) and (390,139.19) .. (390,153) .. controls (390,166.81) and (378.81,178) .. (365,178) .. controls (351.19,178) and (340,166.81) .. (340,153) -- cycle ;
\draw    (324.5,185.5) -- (343.97,173.11) ;
\draw [shift={(346.5,171.5)}, rotate = 147.53] [fill={rgb, 255:red, 0; green, 0; blue, 0 }  ][line width=0.08]  [draw opacity=0] (8.93,-4.29) -- (0,0) -- (8.93,4.29) -- cycle    ;
\draw   (540,153) .. controls (540,139.19) and (551.19,128) .. (565,128) .. controls (578.81,128) and (590,139.19) .. (590,153) .. controls (590,166.81) and (578.81,178) .. (565,178) .. controls (551.19,178) and (540,166.81) .. (540,153) -- cycle ;
\draw    (524.5,185.5) -- (543.97,173.11) ;
\draw [shift={(546.5,171.5)}, rotate = 147.53] [fill={rgb, 255:red, 0; green, 0; blue, 0 }  ][line width=0.08]  [draw opacity=0] (8.93,-4.29) -- (0,0) -- (8.93,4.29) -- cycle    ;
\draw    (165,105) -- (165,125) ;
\draw [shift={(165,128)}, rotate = 270] [fill={rgb, 255:red, 0; green, 0; blue, 0 }  ][line width=0.08]  [draw opacity=0] (8.93,-4.29) -- (0,0) -- (8.93,4.29) -- cycle    ;
\draw    (265,105) -- (265,125) ;
\draw [shift={(265,128)}, rotate = 270] [fill={rgb, 255:red, 0; green, 0; blue, 0 }  ][line width=0.08]  [draw opacity=0] (8.93,-4.29) -- (0,0) -- (8.93,4.29) -- cycle    ;
\draw    (365,105) -- (365,125) ;
\draw [shift={(365,128)}, rotate = 270] [fill={rgb, 255:red, 0; green, 0; blue, 0 }  ][line width=0.08]  [draw opacity=0] (8.93,-4.29) -- (0,0) -- (8.93,4.29) -- cycle    ;
\draw    (565,105) -- (565,125) ;
\draw [shift={(565,128)}, rotate = 270] [fill={rgb, 255:red, 0; green, 0; blue, 0 }  ][line width=0.08]  [draw opacity=0] (8.93,-4.29) -- (0,0) -- (8.93,4.29) -- cycle    ;
\draw    (190,153) -- (237,153) ;
\draw [shift={(240,153)}, rotate = 180] [fill={rgb, 255:red, 0; green, 0; blue, 0 }  ][line width=0.08]  [draw opacity=0] (8.93,-4.29) -- (0,0) -- (8.93,4.29) -- cycle    ;
\draw    (290,153) -- (337,153) ;
\draw [shift={(340,153)}, rotate = 180] [fill={rgb, 255:red, 0; green, 0; blue, 0 }  ][line width=0.08]  [draw opacity=0] (8.93,-4.29) -- (0,0) -- (8.93,4.29) -- cycle    ;
\draw    (390,153) -- (437,153) ;
\draw [shift={(440,153)}, rotate = 180] [fill={rgb, 255:red, 0; green, 0; blue, 0 }  ][line width=0.08]  [draw opacity=0] (8.93,-4.29) -- (0,0) -- (8.93,4.29) -- cycle    ;
\draw    (490,153) -- (537,153) ;
\draw [shift={(540,153)}, rotate = 180] [fill={rgb, 255:red, 0; green, 0; blue, 0 }  ][line width=0.08]  [draw opacity=0] (8.93,-4.29) -- (0,0) -- (8.93,4.29) -- cycle    ;
\draw  [fill={rgb, 255:red, 0; green, 0; blue, 0 }  ,fill opacity=1 ] (449,153) .. controls (449,152.45) and (449.45,152) .. (450,152) .. controls (450.55,152) and (451,152.45) .. (451,153) .. controls (451,153.55) and (450.55,154) .. (450,154) .. controls (449.45,154) and (449,153.55) .. (449,153) -- cycle ;
\draw  [fill={rgb, 255:red, 0; green, 0; blue, 0 }  ,fill opacity=1 ] (459,153) .. controls (459,152.45) and (459.45,152) .. (460,152) .. controls (460.55,152) and (461,152.45) .. (461,153) .. controls (461,153.55) and (460.55,154) .. (460,154) .. controls (459.45,154) and (459,153.55) .. (459,153) -- cycle ;
\draw  [fill={rgb, 255:red, 0; green, 0; blue, 0 }  ,fill opacity=1 ] (469,153) .. controls (469,152.45) and (469.45,152) .. (470,152) .. controls (470.55,152) and (471,152.45) .. (471,153) .. controls (471,153.55) and (470.55,154) .. (470,154) .. controls (469.45,154) and (469,153.55) .. (469,153) -- cycle ;
\draw  [fill={rgb, 255:red, 0; green, 0; blue, 0 }  ,fill opacity=1 ] (479,153) .. controls (479,152.45) and (479.45,152) .. (480,152) .. controls (480.55,152) and (481,152.45) .. (481,153) .. controls (481,153.55) and (480.55,154) .. (480,154) .. controls (479.45,154) and (479,153.55) .. (479,153) -- cycle ;

\draw   (40,153) .. controls (40,139.19) and (51.19,128) .. (65,128) .. controls (78.81,128) and (90,139.19) .. (90,153) .. controls (90,166.81) and (78.81,178) .. (65,178) .. controls (51.19,178) and (40,166.81) .. (40,153) -- cycle ;
\draw    (90,153) -- (137,153) ;
\draw [shift={(140,153)}, rotate = 180] [fill={rgb, 255:red, 0; green, 0; blue, 0 }  ][line width=0.08]  [draw opacity=0] (8.93,-4.29) -- (0,0) -- (8.93,4.29) -- cycle    ;

\draw (107,182.4) node [anchor=north west][inner sep=0.75pt]    {$X_{1}$};
\draw (207,182.4) node [anchor=north west][inner sep=0.75pt]    {$X_{2}$};
\draw (254.5,82.4) node [anchor=north west][inner sep=0.75pt]    {$R_{2}$};
\draw (154.5,82.4) node [anchor=north west][inner sep=0.75pt]    {$R_{1}$};
\draw (252,144.4) node [anchor=north west][inner sep=0.75pt]    {$M_{2}$};
\draw (152,144.4) node [anchor=north west][inner sep=0.75pt]    {$M_{1}$};
\draw (307,182.4) node [anchor=north west][inner sep=0.75pt]    {$X_{3}$};
\draw (354.5,82.4) node [anchor=north west][inner sep=0.75pt]    {$R_{3}$};
\draw (352,144.4) node [anchor=north west][inner sep=0.75pt]    {$M_{3}$};
\draw (507,182.4) node [anchor=north west][inner sep=0.75pt]    {$X_{n}$};
\draw (554.5,82.4) node [anchor=north west][inner sep=0.75pt]    {$R_{n}$};
\draw (552,144.4) node [anchor=north west][inner sep=0.75pt]    {$M_{n}$};
\draw (52,144.4) node [anchor=north west][inner sep=0.75pt]    {$M_{0}$};

\end{tikzpicture}
    \caption{A depiction of the evolution of the memory state of a randomized streaming algorithm $\cA$.
    $X_1, \dots, X_n$ are the inputs the streaming algorithm receives and $R_1, \dots, R_n$ are the independent random bits $\cA$ uses as it processes the stream.
    $M_0$ is the (possibly random) initial state and $M_i$ is the state of $\cA$ after processing $X_i$.
    Note that $M_i$ is a deterministic function of the previous memory state $M_{i-1}$, the $i$th input $X_i$, and the random bits $R_i$. 
    When the inputs are random variables, this figure also depicts the dependence structure of the joint distribution of the random variables $(X_1, \dots, X_n, R_1, \dots, R_n,  M_0, M_1, \dots, M_n)$.}
    \label{fig:general-stream}
\end{figure}
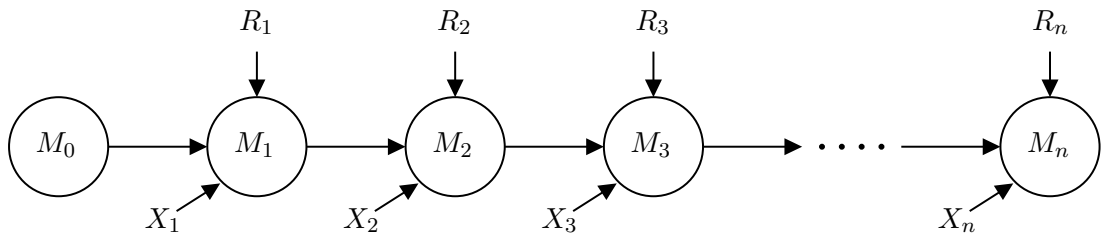

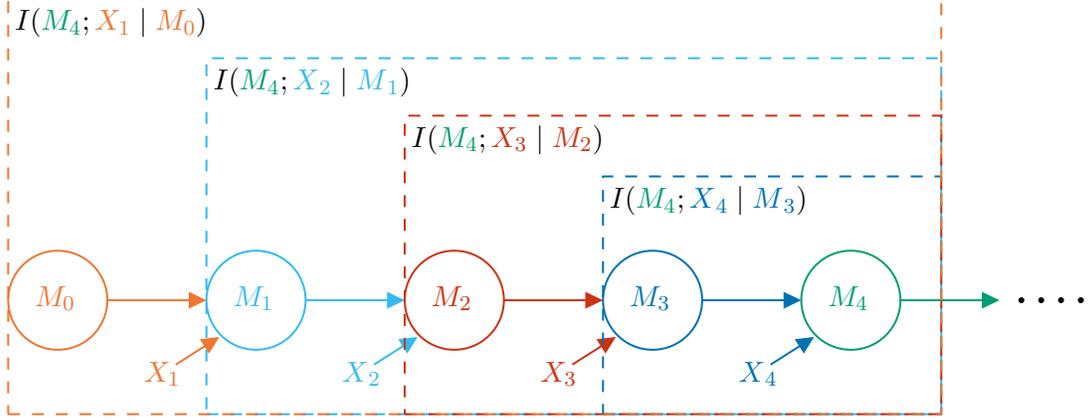
\begin{figure}[H]
    \centering
    \tikzset{every picture/.style={line width=0.75pt}} 

\begin{tikzpicture}[x=0.75pt,y=0.75pt,yscale=-1,xscale=1]

\draw  [color={rgb, 255:red, 51; green, 187; blue, 238 }  ,draw opacity=1 ] (160,173) .. controls (160,159.19) and (171.19,148) .. (185,148) .. controls (198.81,148) and (210,159.19) .. (210,173) .. controls (210,186.81) and (198.81,198) .. (185,198) .. controls (171.19,198) and (160,186.81) .. (160,173) -- cycle ;
\draw [color={rgb, 255:red, 238; green, 119; blue, 51 }  ,draw opacity=1 ]   (144.5,205.5) -- (163.97,193.11) ;
\draw [shift={(166.5,191.5)}, rotate = 147.53] [fill={rgb, 255:red, 238; green, 119; blue, 51 }  ,fill opacity=1 ][line width=0.08]  [draw opacity=0] (8.93,-4.29) -- (0,0) -- (8.93,4.29) -- cycle    ;
\draw  [color={rgb, 255:red, 204; green, 51; blue, 17 }  ,draw opacity=1 ] (260,173) .. controls (260,159.19) and (271.19,148) .. (285,148) .. controls (298.81,148) and (310,159.19) .. (310,173) .. controls (310,186.81) and (298.81,198) .. (285,198) .. controls (271.19,198) and (260,186.81) .. (260,173) -- cycle ;
\draw [color={rgb, 255:red, 51; green, 187; blue, 238 }  ,draw opacity=1 ]   (244.5,205.5) -- (263.97,193.11) ;
\draw [shift={(266.5,191.5)}, rotate = 147.53] [fill={rgb, 255:red, 51; green, 187; blue, 238 }  ,fill opacity=1 ][line width=0.08]  [draw opacity=0] (8.93,-4.29) -- (0,0) -- (8.93,4.29) -- cycle    ;
\draw  [color={rgb, 255:red, 0; green, 114; blue, 178 }  ,draw opacity=1 ] (360,173) .. controls (360,159.19) and (371.19,148) .. (385,148) .. controls (398.81,148) and (410,159.19) .. (410,173) .. controls (410,186.81) and (398.81,198) .. (385,198) .. controls (371.19,198) and (360,186.81) .. (360,173) -- cycle ;
\draw [color={rgb, 255:red, 204; green, 51; blue, 17 }  ,draw opacity=1 ]   (344.5,205.5) -- (363.97,193.11) ;
\draw [shift={(366.5,191.5)}, rotate = 147.53] [fill={rgb, 255:red, 204; green, 51; blue, 17 }  ,fill opacity=1 ][line width=0.08]  [draw opacity=0] (8.93,-4.29) -- (0,0) -- (8.93,4.29) -- cycle    ;
\draw  [color={rgb, 255:red, 0; green, 158; blue, 115 }  ,draw opacity=1 ] (460,173) .. controls (460,159.19) and (471.19,148) .. (485,148) .. controls (498.81,148) and (510,159.19) .. (510,173) .. controls (510,186.81) and (498.81,198) .. (485,198) .. controls (471.19,198) and (460,186.81) .. (460,173) -- cycle ;
\draw [color={rgb, 255:red, 0; green, 114; blue, 178 }  ,draw opacity=1 ]   (444.5,205.5) -- (463.97,193.11) ;
\draw [shift={(466.5,191.5)}, rotate = 147.53] [fill={rgb, 255:red, 0; green, 114; blue, 178 }  ,fill opacity=1 ][line width=0.08]  [draw opacity=0] (8.93,-4.29) -- (0,0) -- (8.93,4.29) -- cycle    ;
\draw [color={rgb, 255:red, 51; green, 187; blue, 238 }  ,draw opacity=1 ]   (210,173) -- (257,173) ;
\draw [shift={(260,173)}, rotate = 180] [fill={rgb, 255:red, 51; green, 187; blue, 238 }  ,fill opacity=1 ][line width=0.08]  [draw opacity=0] (8.93,-4.29) -- (0,0) -- (8.93,4.29) -- cycle    ;
\draw [color={rgb, 255:red, 204; green, 51; blue, 17 }  ,draw opacity=1 ]   (310,173) -- (357,173) ;
\draw [shift={(360,173)}, rotate = 180] [fill={rgb, 255:red, 204; green, 51; blue, 17 }  ,fill opacity=1 ][line width=0.08]  [draw opacity=0] (8.93,-4.29) -- (0,0) -- (8.93,4.29) -- cycle    ;
\draw [color={rgb, 255:red, 0; green, 114; blue, 178 }  ,draw opacity=1 ]   (410,173) -- (457,173) ;
\draw [shift={(460,173)}, rotate = 180] [fill={rgb, 255:red, 0; green, 114; blue, 178 }  ,fill opacity=1 ][line width=0.08]  [draw opacity=0] (8.93,-4.29) -- (0,0) -- (8.93,4.29) -- cycle    ;
\draw  [color={rgb, 255:red, 51; green, 187; blue, 238 }  ,draw opacity=1 ][dash pattern={on 4.5pt off 4.5pt}] (160,50.83) -- (530.6,50.83) -- (530.6,230.5) -- (160,230.5) -- cycle ;
\draw  [color={rgb, 255:red, 204; green, 51; blue, 17 }  ,draw opacity=1 ][dash pattern={on 4.5pt off 4.5pt}] (260,79.83) -- (530.6,79.83) -- (530.6,230.5) -- (260,230.5) -- cycle ;
\draw [color={rgb, 255:red, 0; green, 158; blue, 115 }  ,draw opacity=1 ]   (510,173) -- (557,173) ;
\draw [shift={(560,173)}, rotate = 180] [fill={rgb, 255:red, 0; green, 158; blue, 115 }  ,fill opacity=1 ][line width=0.08]  [draw opacity=0] (8.93,-4.29) -- (0,0) -- (8.93,4.29) -- cycle    ;
\draw  [color={rgb, 255:red, 0; green, 114; blue, 178 }  ,draw opacity=1 ][dash pattern={on 4.5pt off 4.5pt}] (360,110.5) -- (530.6,110.5) -- (530.6,230.5) -- (360,230.5) -- cycle ;
\draw  [fill={rgb, 255:red, 0; green, 0; blue, 0 }  ,fill opacity=1 ] (570,173) .. controls (570,172.45) and (570.45,172) .. (571,172) .. controls (571.55,172) and (572,172.45) .. (572,173) .. controls (572,173.55) and (571.55,174) .. (571,174) .. controls (570.45,174) and (570,173.55) .. (570,173) -- cycle ;
\draw  [fill={rgb, 255:red, 0; green, 0; blue, 0 }  ,fill opacity=1 ] (580,173) .. controls (580,172.45) and (580.45,172) .. (581,172) .. controls (581.55,172) and (582,172.45) .. (582,173) .. controls (582,173.55) and (581.55,174) .. (581,174) .. controls (580.45,174) and (580,173.55) .. (580,173) -- cycle ;
\draw  [fill={rgb, 255:red, 0; green, 0; blue, 0 }  ,fill opacity=1 ] (590,173) .. controls (590,172.45) and (590.45,172) .. (591,172) .. controls (591.55,172) and (592,172.45) .. (592,173) .. controls (592,173.55) and (591.55,174) .. (591,174) .. controls (590.45,174) and (590,173.55) .. (590,173) -- cycle ;
\draw  [fill={rgb, 255:red, 0; green, 0; blue, 0 }  ,fill opacity=1 ] (600,173) .. controls (600,172.45) and (600.45,172) .. (601,172) .. controls (601.55,172) and (602,172.45) .. (602,173) .. controls (602,173.55) and (601.55,174) .. (601,174) .. controls (600.45,174) and (600,173.55) .. (600,173) -- cycle ;

\draw  [color={rgb, 255:red, 238; green, 119; blue, 51 }  ,draw opacity=1 ] (60,173) .. controls (60,159.19) and (71.19,148) .. (85,148) .. controls (98.81,148) and (110,159.19) .. (110,173) .. controls (110,186.81) and (98.81,198) .. (85,198) .. controls (71.19,198) and (60,186.81) .. (60,173) -- cycle ;
\draw [color={rgb, 255:red, 238; green, 119; blue, 51 }  ,draw opacity=1 ]   (110,173) -- (157,173) ;
\draw [shift={(160,173)}, rotate = 180] [fill={rgb, 255:red, 238; green, 119; blue, 51 }  ,fill opacity=1 ][line width=0.08]  [draw opacity=0] (8.93,-4.29) -- (0,0) -- (8.93,4.29) -- cycle    ;
\draw  [color={rgb, 255:red, 238; green, 119; blue, 51 }  ,draw opacity=1 ][dash pattern={on 4.5pt off 4.5pt}] (60,20.83) -- (530.6,20.83) -- (530.6,230.5) -- (60,230.5) -- cycle ;

\draw (127,202.4) node [anchor=north west][inner sep=0.75pt]  [color={rgb, 255:red, 238; green, 119; blue, 51 }  ,opacity=1 ]  {$X_{1}$};
\draw (227,202.4) node [anchor=north west][inner sep=0.75pt]  [color={rgb, 255:red, 0; green, 119; blue, 187 }  ,opacity=1 ]  {$\textcolor[rgb]{0.2,0.73,0.93}{X}\textcolor[rgb]{0.2,0.73,0.93}{_{2}}$};
\draw (272,164.4) node [anchor=north west][inner sep=0.75pt]  [color={rgb, 255:red, 204; green, 51; blue, 17 }  ,opacity=1 ]  {$M_{2}$};
\draw (172,164.4) node [anchor=north west][inner sep=0.75pt]  [color={rgb, 255:red, 51; green, 187; blue, 238 }  ,opacity=1 ]  {$M_{1}$};
\draw (327,202.4) node [anchor=north west][inner sep=0.75pt]  [color={rgb, 255:red, 204; green, 51; blue, 17 }  ,opacity=1 ]  {$X_{3}$};
\draw (372,164.4) node [anchor=north west][inner sep=0.75pt]  [color={rgb, 255:red, 0; green, 114; blue, 178 }  ,opacity=1 ]  {$M_{3}$};
\draw (427,202.4) node [anchor=north west][inner sep=0.75pt]    {$\textcolor[rgb]{0,0.45,0.7}{X}\textcolor[rgb]{0,0.45,0.7}{_{4}}$};
\draw (472,164.4) node [anchor=north west][inner sep=0.75pt]  [color={rgb, 255:red, 102; green, 51; blue, 51 }  ,opacity=1 ]  {$\textcolor[rgb]{0,0.62,0.45}{M_{4}}$};
\draw (362,113.9) node [anchor=north west][inner sep=0.75pt]    {$I(\textcolor[rgb]{0,0.62,0.45}{M_{4}} ;\textcolor[rgb]{0,0.45,0.7}{X}\textcolor[rgb]{0,0.45,0.7}{_{4}} \mid \textcolor[rgb]{0,0.45,0.7}{M}\textcolor[rgb]{0,0.45,0.7}{_{3}})$};
\draw (262,83.23) node [anchor=north west][inner sep=0.75pt]    {$I(\textcolor[rgb]{0,0.62,0.45}{M_{4}} ;\textcolor[rgb]{0.8,0.2,0.07}{X_{3}} \mid \textcolor[rgb]{0.8,0.2,0.07}{M_{2}})$};
\draw (162,54.23) node [anchor=north west][inner sep=0.75pt]    {$I(\textcolor[rgb]{0,0.62,0.45}{M_{4}} ;\textcolor[rgb]{0.2,0.73,0.93}{X}\textcolor[rgb]{0.2,0.73,0.93}{_{2}} \mid \textcolor[rgb]{0.2,0.73,0.93}{M}\textcolor[rgb]{0.2,0.73,0.93}{_{1}})$};
\draw (72,164.4) node [anchor=north west][inner sep=0.75pt]  [color={rgb, 255:red, 51; green, 187; blue, 238 }  ,opacity=1 ]  {$\textcolor[rgb]{0.93,0.47,0.2}{M_{0}}$};
\draw (62,24.23) node [anchor=north west][inner sep=0.75pt]    {$I(\textcolor[rgb]{0,0.62,0.45}{M_{4}} ;\textcolor[rgb]{0.93,0.47,0.2}{X_{1}} \mid \textcolor[rgb]{0.93,0.47,0.2}{M_{0}})$};

\end{tikzpicture}
    \caption{An illustration of a single inner sum of the information cost $\sum_{j \le i} I(M_i ; X_j | M_{j-1})$ for $i = 4$.
    To simplify the presentation we do not include the random bits that the streaming algorithm uses while processing the stream in this illustration.}
    \label{fig:info-cost}
\end{figure}

As it was shown in~\cite{BravermanGW20}, this quantity lower bounds $\sum_{i=1}^T \left|M_i\right|$, i.e., $T$ times the memory consumption, and it satisfies the direct-sum property: solving $k$ independent instances of the problem needs exactly $k$ times $\IC$.
Thus, it suffices to prove an $\IC$ lower bound of $\Omega(T\log T)=\Omega(T\log(n/k))$ for a single approximate counter.
Intuitively, such a lower bound means that the algorithm must spend $\Omega(1)$ bits remembering each bit of the sum (recall that we ensure $\sum X_i\leq \poly\, T$).

Let us first focus on the lowest bit, i.e., the parity of the sum.
Note that one does not have to know the lowest bit in order to return an approximation of the sum.
Nevertheless, we will show that the algorithm has to constantly pay attention to this bit in order to output a good approximation with very high probability.
To this end, let us consider the distribution where the $X_i$'s are independent and uniform in $\{0,1\}$.
Let us focus on the terms in $\IC$ with $j=i$, i.e., those of the form $I(M_i; X_i\mid M_{i-1})$.
Intuitively, under this distribution, this term represents how much attention the algorithm is paying to the evolution of the parity at batch $i$.
This is because $X_i$ is simply the difference between the parities of the first $i-1$ and first $i$ inputs.

Suppose that $I(M_i; X_i\mid M_{i-1})\ll 1$ for a constant fraction of $i$. Then, we can show that the algorithm will make an error of $\Omega(T)$ with at least $\delta$ probability.
Roughly speaking, this is because for each such $i$, conditioned on $M_i$ and $M_{i-1}$, the distribution of $X_i$ is still close to uniform (as $X_i$ is uniform conditioned on $M_{i-1}$).
Therefore, if we condition on all $M_0,M_1,\ldots,M_T$, most $X_i$ can still take both values $0$ or $1$ with constant probability, and all $X_i$ are still independent by the Markov property of the algorithm.
In particular, by setting all these $X_i$'s to $0$ or setting all to $1$, we reach the same final memory state $M_T$, but in the two cases, the total sum differs by $\Omega(T)$.
Since both happen with probability $2^{-O(T)}\gg \delta$ given the final memory state $M_T$, the algorithm must make an error of at least $\Omega(T)$ with probability at least $\delta$.

We can extend this argument to any specific bit of the sum by considering a stream with $\Theta(T/B^l)$ independent random increments that are uniform in $\{0,B^l\}$ for some constant $B$ and $l\in[L]$, where $L=\log_B T$.
Our final hard distribution interleaves $L$ such streams, which we call the scales.
For each scale $l\in[L]$, we evenly spread the $\Theta(T/B^l)$ random increments in the whole stream with a gap of $\Theta(B^l)$ batches.
Note that now the sum of all inputs is always at most $O(T\log T)$.
For the sum in the definition of $\IC$, we will only focus on $L$ terms for each $i$: $I(M_i; X_{\lfloor i\rfloor_l} \mid M_{\lfloor i\rfloor_l-1})$, where $X_{\lfloor i\rfloor_l}$ is the closest scale $l$ input before $M_i$.
If $\IC\ll T\cdot L=O(T\log T)$, then there must exist some scale $l^*$ such that $I(M_i; X_{\lfloor i\rfloor_{l^*}} \mid M_{{\lfloor i\rfloor_{l^*}}-1})\ll 1$ for most $i$.
Then we apply the above argument, and show that conditioned on the memory states right before each scale-$l^*$ input, most scale-$l^*$ inputs can still take values $0$ or $B^{l^*}$ with constant probability, and the scale-$l^*$ inputs are still independent. 
We further observe that for most of them, the inputs between $X_{\lfloor i\rfloor_{l^*}}$ and $M_i$ are coming from the lower scales $l<l^*$.
The standard deviation of their sum is much smaller than $B^{l^*}$, and we can show that they do not affect the sum by too much as we alter $X_{\lfloor i\rfloor_{l^*}}$.
Thus, by setting these scale-$l^*$ inputs to $0$ or $B^{l^*}$, the entire sum will again differ by $\Omega(T)$, but they lead to the same final memory state, i.e., the algorithm does not distinguish between the two cases.
Since the sum is $O(T\log T)$, such a difference is more than $\eps$ times the sum when $\eps<O(1/\log T)$. A more careful analysis shows that this happens with probability at least $2^{-O(T)}>\delta$, leading to a contradiction.

\paragraph{Discussion on the choice of the hard distribution.} We note here that the independent uniform $\{0,1\}$ distribution, by itself, is \emph{not} hard for the information cost.
One solution with low information cost is to divide $X_i$'s into blocks of size $S=\Theta(\eps^{-2}\log T)$, and maintain the exact sum within the current block.
If \emph{all} blocks have sums $(1\pm \eps)S/2$, then we simply remember this fact and use $T/2$ as the final output.
Otherwise, the algorithm finds a block whose sum is not in this range, then it maintains the exact sum for all future blocks.
Since all previous blocks have sums $(1\pm\eps)S/2$, in particular, $S/2$ is a $(1 + \eps)$-approximation for them, the algorithm can also return a $(1+ \eps)$-approximation of the whole sum.
Now note that this case only happens with $1/\poly(T)$ probability, and the total information cost is at most $T$ times the expected memory usage.
The expected memory is at most $O(\log(1/\eps)+\log\log T)$ bits for maintaining the sum in the current block, plus $O(T^{-\Theta(1)}\cdot \log T)$ bits in expectation for maintaining the entire exact sum.
When $\eps=\Theta(1/\log T)$, the information cost is only $O(T\log\log T)\ll T\log T$.

The above strategy works since the sum of block is concentrated around the expectation, hence, we need extra space to maintain the exact sum only with very small probability.
One can also show that the above strategy also applies to any i.i.d. distributions with some concentration.\footnote{For example, it suffices to have finite $\E[|X|^{1.01}]$.}
Our hard distribution is inspired by the \emph{discrete half-Cauchy distribution}, which has probability $\Theta(1/(x+1)^2)$ at integers $x\geq 0$.
This distribution does not have an expectation, hence, there is no concentration for blocks of any size.
We also have a more involved proof of our main result that uses this distribution instead; the main property that proof relies on is that for every $W$ random variables, roughly one of them takes value $\Theta(W)$.
The hard distribution we actually use in this paper is a bounded distribution that can be viewed as extracting this useful property of the half-Cauchy distribution, but which can be analyzed more simply.

\section{Lower Bounds}

In this section we will prove our lower bound for the space complexity of $(k,\eps,\delta)$-approximate counting.
We split the proof into two parts.
We first focus on proving the difficult part of the lower bound that depends on the failure probability $\delta$.
To do so, we derive a lower bound for the space complexity of $(\eps,\delta)$-approximate counting by appealing to an information theoretic argument.
By using a good definition of information cost together with an appropriately chosen hard distribution, we can prove that any accurate algorithm remember a lot of information about many different parts of the stream, i.e.\ it incurs a high information cost. 
This immediately gives us a lower bound on the memory size.
We then use this result to prove a space lowerbound for $(k,\eps,\delta)$-approximate counting via a direct sum argument.
We conclude the section by proving the portion of the lower bound that depends on the approximation error $\eps$ and the total sum of all counters $n$ by generalizing the argument used in the single counter case.

\subsection{Information lower bound for a single counter}

Recall that for a randomized streaming algorithm $\cA$ we define $M_i$ to be the state of $\cA$ after processing the $i$th input $X_i$ together with some additional independent random bits $R_i$, i.e. $M_i$ is a deterministic function of $M_{i-1}$, $X_i$ and $R_i$ (equivalently, a deterministic function of $X_{\le i}$ and $R_{\le i}$).
The following is a notion of information cost for streaming algorithms originally defined by Braverman, Garg, and Woodruff~\cite{BravermanGW20}.
\begin{definition}\label{def:information_cost}
Let $\cA$ be a randomized streaming algorithm.
Given a distribution $\cD$ over input sequences of length $s$, we define the \emph{information cost} of algorithm $\cA$ on input $(X_1, \dots, X_s) \sim \cD$ to be
\[
\IC(\cA , \cD) := \sum_{i=1}^{s}\sum_{j=1}^{i} I(M_i; X_{j} \mid M_{j-1}).
\]
\end{definition}
The above definition of the information cost is motivated by the following chain of inequalities:  
\begin{align*}
  \E |M_i| &\ge H(M_i) && \text{(source coding theorem\footnotemark )}\\
  {}&\ge I(M_i ; X_{\le i}, M_{<i})\\
  {}&= \sum_{j=1}^i I(M_i; X_j, M_{j-1} \mid X_{<j}, M_{<j-1})\\
  {}&= \sum_{j=1}^i I(M_i;M_{j-1}\mid X_{<j},M_{<j-1}) + I(M_i;X_j\mid X_{<j},M_{<j})\\
  {}&= \sum_{j=1}^i I(M_i;X_j\mid X_{<j},M_{<j})\\
  {}&= \sum_{j=1}^i I(M_i;X_j\mid M_{j-1}). 
\end{align*}
\footnotetext{The source coding theorem holds for any prefix code. In general we may lose a factor of $2$ in this inequality: we have both $\E[|M_i|]\ge H(M_i\mid |M_i|)$ and $\E[|M_i|]=\sum_{n\ge 1} np_n=\sum_{n\ge 1} p_n\log(1/2^{-n})=\sum_{n\ge 1} p_n\log(p_n/2^{-n})+H(|M_i|)\ge H(|M_i|)$ for $p_n:=\bP(|M_i|=n)$; consequently $2\E[|M_i|]\ge H(M_i\mid |M_i|)+H(|M_i|)\ge H(M_i)$.}
This implies that $\IC(\cA; \cD)\le \sum_{i=1}^s \E |M_i|$. 
The main technical part of this paper is proving the following lower bound for a single counter using this notion of information cost.
\begin{lemma}\label{lem:single_information_lowerbound}
Let $\cA$ be a $(\eps,\delta)$-approximate counter with parameters $\delta \in (0,c_1)$ and $\epsilon \in (0, \frac{c_2}{\log\log(1/\delta)})$ that uses $|M|$ bits of space. 
There is a distribution $\cD$ over inputs such that the information cost of $\cA$ on $(X_1, \dots, X_n) \sim \cD$ satisfies
\[
\IC(\cA ; \cD) = \Omega(n \log\log (1/\delta) ),
\]
where $\bP[\sum_{i=1}^n X_i \le n] = 1$  and $n \ge c_3\log^{c_4}(1/\delta)$.
This implies the space lower bound
\[
\frac{1}{n} \sum_{i=1}^{n} \bE |M_i| = \Omega(\log\log(1/\delta)) = \Omega(\min\{\log n, \log\log(1/\delta)\}).
\]
\end{lemma}
\begin{proof}
We construct the distribution $\cD$ as follows: let $B\in \mathbb{N}$ be a large integer constant to be specified later, and $T$ is the largest power of $B$ such that $ T < \log_{32}(1/\delta)$. When $\delta < c_1$ where $c_1 = c_1(B)$ is a sufficiently small constant, we have $T \ge B$ and so $L := \log_B T \in \mathbb{N}$. The distribution $\cD$ is an interleaving of $L$ distributions $\cD_\ell$ on $L$ different scales, where for each scale $\ell\in [L]$, the distribution $\cD_\ell$ is a product distribution $\prod_{i=1}^{ T }\cD_{\ell, i}$: 
\begin{align*}
\text{ under } \cD_{\ell, i}, \quad Y_{\ell, i} \begin{cases}
\sim \text{Unif}(\{0, B^\ell \}) &\text{if } B^\ell \text{ divides } (i-1), \\
= 0 &\text{otherwise.}
\end{cases} 
\end{align*}
For notational simplicity, we also denote the non-zero entries of $Y_{\ell,i}$ by $Z_{\ell, j} = Y_{\ell, B^{\ell}(j-1)+1}$ for $j\in [T/B^\ell]$.
The stream under distribution $\cD$ is then generated by interleaving the $Y_{\ell, i}$ terms to form the sequence $X^{ n }=(Y_{1,1},\cdots, Y_{L,1}, Y_{1,2}, \cdots, Y_{L,2}, \cdots Y_{L,T })$, where $n := TL$ is the length of the sequence.
The total value of the counter under this stream is at most
\begin{align*}
\sum_{\ell=1}^L B^\ell\cdot \frac{T}{B^{\ell}} = n. 
\end{align*}
Finally, note that $n = TL \ge c_3\log^{c_4}(1/\delta)$ for an appropriate choice of constants $c_3, c_4 \ge 0$.

Let $\cA$ be an $(\eps,\delta)$-approximate counter, and assume towards contradiction that $\IC(\cA;\cD) < \alpha n L$ where $\alpha$ is a small constant.
For any index $i\in [n]$ in the stream, let $\rndd{i} := B^\ell L \lfloor(i-1)/(B^\ell L) \rfloor+\ell$ be the index of the closest $Z_{\ell, j}$ to the left of $X_i$ in the stream, i.e.\ $X_{\rndd{i}} = Z_{\ell, j}$.
Similarly, for $j\in [T/B^{\ell}]$, we define $\rndu{j} := B^{\ell}L(j-1) + \ell$ to be the index of the $j$th non-zero entry of scale $\ell$ in the stream, i.e.\ $X_{\rndu{j}} = Z_{\ell,j}$.
The definition of the information cost tells us that
\begin{align*}
    \alpha n L  &> \IC(\cA; \cD) \\
    &= \sum_{i=1}^n \sum_{j\le i} I(M_i; X_j \mid M_{j-1}) \\
    &\ge \sum_{i=1}^n \sum_{\ell=1}^L I(M_i; X_{\rndd{i}} \mid M_{\rndd{i}-1}) \\
    &= \sum_{\ell=1}^L \sum_{j=1}^{T/B^\ell}\sum_{i = \rndu{j}}^{\rndu{j+1} - 1} I(M_i; X_{\rndu{j}} \mid M_{\rndu{j} -1}) \\
    &\ge \sum_{\ell=1}^L \sum_{j=1}^{T/B^\ell} B^{\ell}L\cdot I(M_{\rndu{j+1} - 1}; X_{\rndu{j}} \mid M_{\rndu{j} -1}) \\
    &\ge \min_{\ell\in [L]} B^{\ell} L^{2} \sum_{j=1}^{T/B^\ell} I(M_{\rndu{j+1} - 1}; X_{\rndu{j}}\mid M_{\rndu{j} -1})\\
    &=  \min_{\ell\in [L]} B^{\ell} L^{2}  \sum_{j=1}^{T/B^\ell} I(M_{\rndu{j+1} - 1}; Z_{\ell,j}\mid M_{\rndu{j} -1})
\end{align*}
where the second last inequality is due to the data processing inequality. 
By Markov's inequality, there exist $\ell_0\in [L]$ and $J_0\subseteq [ T /B^{\ell_0}]$ with $|J_0|\ge T /(2B^{\ell_0})$ and 
\begin{align*}
I(M_{\lceil j+1 \rceil_{\ell_0}-1}; Z_{\ell_0, j} \mid M_{\lceil j \rceil_{\ell_0}-1})\le 2\alpha, \qquad \forall j\in J_0. 
\end{align*}

For ease of presentation, we abuse notation and write the above inequality as $I(M_j; Z_j \mid M_{j-1})\le 2\alpha$ for $j\in J_0$. We shall also keep in mind that the stream between $M_{j-1}$ and $M_j$ contains $(Z^{<}_{j}, Z_j, Z^{>}_{j})$, where $Z^{<}_{j}$ is the collection of all non-zero inputs $\{Z_{\ell', j'}\}$ inside this window with a lower scale $\ell'<\ell$, and $Z^{>}_{j}$ is the counterpart with a higher scale $\ell'>\ell$ (see Figure~\ref{fig:interleaving}).

Next we define several good events for the sake of analysis. The first good event $E_{j,1}$ characterizes the behavior of the contribution of $Z^{<}_{j}$ and is formally defined as
\begin{align*}
    \mathbbm{1}(E_{j,1}) := \mathbbm{1}\left( \left|\text{sum}(Z^{<}_{j}) - \bE[\text{sum}(Z^{<}_{j}) \mid M_{j-1}, M_j] \right| \le \frac{B^{\ell_0}}{4} \right), \quad j\in J_0. 
\end{align*}
By Chebyshev's inequality, it is clear that
\begin{align*}
\bE_{M_{j-1},M_j}[ \bP(E_{j,1}^c\mid M_{j-1}, M_j) ] &\le \bE_{M_{j-1},M_j}\left[\frac{\var(\text{sum}(Z^{<}_{j}\mid M_{j-1},M_j) }{(B^{\ell_0}/4)^2} \right] \\
&\le \frac{\var(\text{sum}(Z^{<}_{j}))}{(B^{\ell_0}/4)^2} = \frac{16}{B^{2\ell_0}}\sum_{\ell<\ell_0} B^{\ell_0-\ell} \cdot \frac{B^{2\ell}}{4} \le \frac{4}{B-1}. 
\end{align*}
Consequently, it holds that
\begin{align*}
&\bE_{M_{j-1},M_j}[\TV(P_{Z_j}, P_{Z_j\mid M_{j-1},M_j,E_{j,1}} )] \\
&= \bE_{M_{j-1},M_j}[\TV(P_{Z_j\mid M_{j-1}}, P_{Z_j\mid M_{j-1},M_j,E_{j,1}} )] \\
&\le \bE_{M_{j-1},M_j}[\TV(P_{Z_j\mid M_{j-1}}, P_{Z_j\mid M_{j-1},M_j} )+\TV(P_{Z_j\mid M_{j-1},M_j}, P_{Z_j\mid M_{j-1},M_j,E_{j,1}} )] \\
&= \bE_{M_{j-1},M_j}[\TV(P_{Z_j\mid M_{j-1}}, P_{Z_j\mid M_{j-1},M_j} )] + \bE_{M_{j-1},M_j}[\bP(E_{j,1}^c\mid M_{j-1}, M_j)] \\
&\le \sqrt{\frac{1}{2}\bE_{M_{j-1},M_j}[\KL(P_{Z_j\mid M_{j-1}}\| P_{Z_j\mid M_{j-1},M_j} )] } + \frac{4}{B-1} \\
&= \sqrt{I(Z_j; M_j\mid M_{j-1})} + \frac{4}{B-1} \\
&\le \sqrt{\alpha} + \frac{4}{B-1}, 
\end{align*}
which can be made small by choosing $\alpha>0$ small enough and $B\in \mathbb{N}$ large enough. Note that in the above inequality we have used the triangle inequality $\TV(P,Q)\le \TV(P,R)+\TV(Q,R)$, the conditioning relationship $\TV(P, P_{\mid E}) = P(E^c)$, Pinsker's inequality $\TV(P,Q)\le \sqrt{\KL(P\|Q)/2}$, and Jensen's inequality $\E[\sqrt{X}]\le \sqrt{\E[X]}$. 

The next good event $E_2$ concerns the simultaneous occurrence of $\{E_{j,1}\}$ for a constant proportion of $j\in J_0$. Specifically, $E_2$ is the event that there exists some $J_1\subseteq J_0$ such that:
\begin{enumerate}
    \item $|J_1|\ge |J_0|/2 \ge T /(4B^{\ell_0})$; 
    \item event $E_{j,1}$ is true for all $j\in J_1$; 
    \item a small TV distance $\TV(P_{Z_j}, P_{Z_j\mid M_{j-1},M_j,E_{j,1}} )\le 1/4$ (denoted by event $E_{j,2}$) for all $j\in J_1$. 
\end{enumerate}
Note that $\{(Z^{<}_{j},Z_j,Z^{>}_{j})\}$ are conditionally independent given $\{M_j\}$, we have
\begin{align*}
&\bE_{\{M_j\}} \left[\sum_{j\in J_0} \mathbbm{1}(E_{j,1}\cap E_{j,2}) \right] \\
&= \sum_{j\in J_0} \bE_{M_{j-1},M_j}[\mathbbm{1}(E_{j,1}\cap E_{j,2})] \\
&\ge \sum_{j\in J_0} \left(1 - \bE_{M_{j-1},M_j}[ \bP(E_{j,1}^c\mid M_{j-1}, M_j) ] 
- 4\cdot \bE_{M_{j-1},M_j}[\TV(P_{Z_j}, P_{Z_j\mid M_{j-1},M_j,E_{j,1}} )]\right)\\
&\ge \left(1 - \frac{4}{B-1} - 4\left(\sqrt{\alpha}+\frac{4}{B-1}\right)\right)|J_0| \\
&\ge \frac{3}{4}|J_0|, 
\end{align*}
by choosing $\alpha$ small enough and $B$ large enough. Consequently, by Markov's inequality, we have $\bP(E_2)\ge 1/2$ over the randomness of $\{M_j\}$ and $\{(Z^{<}_{j},Z_j,Z^{>}_{j})\}$. 

Now we condition on $E_2$ and arrive at the desired contradiction. For a probability distribution $P$ over the real line and $\Delta\ge 0$, define a quantity $f(P,\Delta)$ as follows:
\begin{align*}
    f(P,\Delta) = \max\{\delta>0: \exists L\in \mathbb{R} \text{ such that } P((-\infty,L])\ge \delta, P([L+\Delta,\infty)) \ge \delta \}. 
\end{align*}

Intuitively, a \emph{small} $f(P,\Delta)$ implies that the distribution $P$ assigns a lot of probability to \emph{some} interval of length $\Delta$.
The following lemma summarizes some properties of $f(P,\Delta)$. 
\begin{lemma}\label{lemma:f_P_Delta}
Let $P$ and $Q$ be two probability distributions over $\mathbb{R}$, and $P \star Q$ denote their convolution. For $\Delta_1, \Delta_2, \Delta\ge 0$, it holds that
\begin{align*}
f(P\star Q, \Delta_1+\Delta_2) &\ge f(P,\Delta_1)f(Q,\Delta_2), \\
f(P\star Q, \Delta) &\ge f(P, \Delta)/2. 
\end{align*}
\end{lemma}
\begin{proof}
For the first claim, suppose that
\begin{align*}
\min\{P((-\infty,L_1)], P([L_1+\Delta_1,\infty))\} &\ge f(P,\Delta_1), \\
\min\{Q((-\infty,L_2)], Q([L_2+\Delta_2,\infty))\} &\ge f(Q,\Delta_2). 
\end{align*}
Then the first inequality follows from
\begin{align*}
P\star Q((-\infty, L_1+L_2]) &\ge P((-\infty,L_1])Q((-\infty,L_2]) \ge f(P,\Delta_1)f(Q,\Delta_2), \\
P\star Q([L_1+L_2+\Delta_1+\Delta_2,\infty)) &\ge P([L_1+\Delta_1,\infty))Q([L_2+\Delta_2,\infty)) \ge f(P,\Delta_1)f(Q,\Delta_2). 
\end{align*}
The second claim is a direct consequence of the first claim and $f(Q,0)\ge 1/2$. 
\end{proof}

To apply Lemma \ref{lemma:f_P_Delta}, we consider the conditional distribution $P_{S\mid \{M_j\}}$, where $S=\sum_{i=1}^n X_i$ is the total number of counter updates, and $\{M_j\}$ are the memory states at scale $\ell_0$ defined before. Since the counter $\cA$ is $(\eps,\delta)$-approximate, in expectation $P_{S\mid \{M_j\}}$ must have probability mass at least $1-\delta$ on an interval of size $2\eps n$. This implies that 
\begin{align}\label{eq:f_upper_bound}
    \bE_{\{M_j\}}[f(P_{S\mid \{M_j\}}, 2\epsilon n)] \le \delta. 
\end{align}

On the other hand, since $\{(Z^{<}_{j},Z_j,Z^{>}_{j})\}$ are conditionally independent given $\{M_j\}$, we may invoke Lemma \ref{lemma:f_P_Delta} (first part for $J_1$ and second part for $J_1^c$) to arrive at
\begin{align*}
\bE_{\{M_j\}}[f(P_{S\mid \{M_j\}}, 2\epsilon n)] &\ge \bP(E_2)\bE_{\{M_j\}}[f(P_{S\mid \{M_j\}}, 2\epsilon n)\mid E_2] \\
&\ge \frac{1}{2}\cdot \bE_{\{M_j\}}\left[\prod_{j\in J_1} f(P_{\text{sum}(Z^{<}_{j},Z_j,Z^{>}_{j}) \mid M_{j-1},M_j }, \frac{2\eps n}{|J_1|}) \cdot \left(\frac{1}{2}\right)^{\frac{T}{B^{\ell_0}} - |J_1|} \mid E_2\right]  \\
&\ge \frac{1}{2^T}\cdot \bE_{\{M_j\}}\left[\prod_{j\in J_1} f(P_{\text{sum}(Z^{<}_{j},Z_j,Z^{>}_{j}) \mid M_{j-1},M_j }, 8\eps LB^{\ell_0}) \mid E_2\right]. 
\end{align*}
Conditioned on the event $E_2$, the event $E_{j,1}$ implies that the deviation of $\text{sum}(Z^{<}_{j})$ to its posterior mean is at most $B^{\ell_0}/4$, and the event $E_{j,2}$ implies that the posterior (marginal) distribution of $Z_j$ puts at least $1/4$ probability mass on both $0$ and $B^{\ell_0}$. Moreover, $\text{sum}(Z^{>}_{j})$ is always an integral multiple of $B^{\ell_0+1}$. Now we prove that 
\begin{align*}
    f(P_{\text{sum}(Z^{<}_{j},Z_j,Z^{>}_{j}) \mid M_{j-1},M_j,E_{j,1},E_{j,2} }, \frac{B^{\ell_0}}{3}) \ge \frac{1}{16}.
\end{align*}
We distinguish into two cases: 
\begin{enumerate}
    \item Case I: $\bP(\text{sum}(Z^{>}_{j}) \neq \text{median}(\text{sum}(Z^{>}_{j}))) \ge 1/8$. As $\text{sum}(Z^{>}_{j})$ is always an integral multiple of $B^{\ell_0+1}$, this implies $|\text{median}(\text{sum}(Z^{>}_{j})) - \text{sum}(Z^{>}_{j})|\ge B^{\ell_0+1}/2$ with probability at least $1/8$. By symmetry, without loss of generality we may assume that $\text{median}(\text{sum}(Z^{>}_{j})) - \text{sum}(Z^{>}_{j})\ge B^{\ell_0+1}/2$ with probability at least $(1/8)/2=1/16$. 
    
    Moreover, the range of $Z_j$ is $B^{\ell_0}$, and the range of $\text{sum}(Z^{<}_j)$ is at most $B^{\ell_0}/2$ under $E_{j,1}$. Consider the interval
    \begin{align*}
        \left[\text{median}(\text{sum}(Z_j^<, Z_j, Z^{>}_{j}), \text{median}(\text{sum}(Z_j^<, Z_j, Z^{>}_{j})) + \frac{B^{\ell_0}}{3} \right]
    \end{align*}
    of length $B^{\ell_0}/3$, it is clear that
    \begin{align*}
        &\bP\left(\text{sum}(Z_j^<, Z_j, Z^{>}_{j}) \le \text{median}(\text{sum}(Z_j^<, Z_j, Z^{>}_{j}))\right) \ge \frac{1}{2}, \\
        &\bP\left(\text{sum}(Z_j^<, Z_j, Z^{>}_{j}) \ge \text{median}(\text{sum}(Z_j^<, Z_j, Z^{>}_{j})) + \frac{B^{\ell_0}}{3}\right) \\
        &\ge \bP\left(\text{sum}(Z^{>}_{j}) \ge \text{median}(\text{sum}(Z^{>}_{j})) + \frac{3B^{\ell_0}}{2} + \frac{B^{\ell_0}}{3}\right) \ge \frac{1}{16},
    \end{align*}
    as long as $B^{\ell_0+1}/2 \ge 11B^{\ell_0}/6$, or equivalently $B\ge 4$. 
    \item Case II: $\bP(\text{sum}(Z^{>}_{j}) = \text{median}(\text{sum}(Z^{>}_{j}))) \ge 7/8$. In this case, we argue that 
\begin{align*}
\bP\left(\text{sum}(Z^{<}_{j}) + Z_j + \text{sum}(Z^{>}_{j}) \le \bE[\text{sum}(Z^{<}_{j})] + \frac{B^{\ell_0}}{3} + \text{median}(\text{sum}(Z^{>}_{j})) \right) &\ge \frac{1}{16}, \\
\bP\left(\text{sum}(Z^{<}_{j}) + Z_j + \text{sum}(Z^{>}_{j}) \ge \bE[\text{sum}(Z^{<}_{j})] + \frac{2B^{\ell_0}}{3} + \text{median}(\text{sum}(Z^{>}_{j})) \right) &\ge \frac{1}{16}. 
\end{align*}
By symmetry we only prove the first claim, where the event occurs whenever $Z_j=0$ and $\text{sum}(Z^{>}_{j})$ is equal to its median. By the union bound, this happens with probability at least $1 - (1/8 + 3/4) = 1/8$. 
\end{enumerate}

Consequently, as long as $\eps\le 1/(24L)=O(1/\log\log(1/\delta))$, we have 
\begin{align}\label{eq:f_lower_bound}
    \bE_{\{M_j\}}[f(P_{S\mid \{M_j\}}, 2\epsilon n)] \ge \left(\frac{1}{32}\right)^T.
\end{align}
However, as $T<\log_{32}(1/\delta)$, inequalities \eqref{eq:f_upper_bound} and \eqref{eq:f_lower_bound} are contradictory to each other. Thus, the assumption that $\IC(\cA;\cD) < \alpha n L$ cannot be true. 
\end{proof}

\begin{figure}[H]
    \input{interleaving.tikz}
    \caption{A simple example that illustrates the interleaving process for our random stream. We set the parameters $B = 2$ and $T = 8$ so that we sample from $L = 3$ scales to get a stream of total length $n = 24$. We do not illustrate the entire stream to save space. In this example we consider $\ell_0 = 2$ and demonstrate what a typical $Z_j$ looks like. The sequence of random variables inside the red dotted box/window consists of the terms $(Z^{<}_{j}, Z_{j}, Z^{>}_{j} )$ and we highlight the memory states $M_{j-1}, M_j$.}
    \label{fig:interleaving}
\end{figure}

\subsection{Amortized space complexity via direct sum}

We will now ``lift'' the information lower bound for any $(\eps,\delta)$-approximate counter to a lower bound for any $(k,\eps,\delta)$-approximate counter from which we can derive a lower bound on the memory size.
The proof follows from a simple direct sum argument applied to the information cost.
\begin{lemma}\label{lem:multiple_information_lowerbound}
Any $(k,\eps,\delta)$-approximate counter with parameters $\delta \in (0,c_1)$ and $\epsilon \in (0, \frac{c_2}{\log\log(1/\delta)})$ that maintains $k$ counters with total sum at most $n\geq k (c_3 \log^{c_4}(1/\delta))$ must use
\[
\Omega(k\log\log(1/\delta)) = \Omega(k\min\{\log(n/k) , \log\log(1/\delta)\})
\]
bits of memory in expectation.
\end{lemma}
\begin{proof}
We prove the lower bound via a reduction to the information lower bound for a single counter.
That is, we will embed updates from the hard distribution $\cD$ used to prove the single counter lower bound in Lemma~\ref{lem:single_information_lowerbound} into the $k$ counter problem.
Fix a  $(k,\eps, \delta)$-approximate counter  $\cA_k$.
Since the distribution $\cD$ is a product distribution, we can write $\cD = \prod_{i=1}^{n'} \cD_i$ for $n':=n/k$.
We define $X^{k}_{i} := (X_{i,1}, \dots, X_{i,k}) \sim (\mathcal D_{i})^{k}$ for each $i \in [n']$.
We consider $X^{k}_{i}$ to be the $i$th input to $\cA_k$ where $X_{i,j}$ is the update applied to $j$th counter.
The overall input to $\cA_k$ is
\[
(X^{k}_{1}, \dots, X^{k}_{n'}) \sim \prod_{i=1}^{n'}(\cD_{i})^k,
\]
and the total sum of all counts is at most $n = kn' \ge k(c_3\log^{c_4}(1/\delta))$ (see Lemma~\ref{lem:single_information_lowerbound}).
Define $\cD_{\text{int}} := \prod_{i=1}^{n'}(\cD_{i})^k$.
Let $M_{i}$ be the memory state of the algorithm after processing the $i$th input $X^{k}_{i}$.
We can lower bound the information cost incurred by $\cA_k$:
\[
\IC(\cA_k , \cD_{\text{int}}) = \sum_{i=1}^{n'}\sum_{j=1}^{i} I(M_i; X^{k}_{j}\mid M_{j-1}) \ge \sum_{u = 1}^{k} \sum_{i=1}^{n'}\sum_{j=1}^{i} I(M_i; X_{j,u} \mid M_{j-1}),
\]
where the inequality follows from the superadditivity of conditional mutual information (Proposition~\ref{prop:superadd}).

We now explain how we embed the single counter problem into the $k$ counter problem using $\cA_k$.
Given the input $X^{n'} = (X_1, \dots, X_{n'}) \sim \cD$ for a single counter, we pick an index $u \in [k]$ uniformly at random and proceed as if the updates are applied to $u$th counter of $\cA_k$.
Denote by $U$ this uniformly random index.
For the other counters, we simulate ``fake'' updates from the same distribution and apply them to $\cA_k$ as if they were received as inputs.
Denote by $\cA'$ the resulting approximate counter that maintains $X^{n'}$.
We can upper bound the information cost of $\cA'$ by
\begin{align*}
\IC(\cA' , \cD) =  \sum_{i=1}^{n'}\sum_{j=1}^{i} I(M_i; X_{j,U} \mid M_{j-1}, U) &= \frac{1}{k} \sum_{u=1}^{k}\sum_{i=1}^{n'}\sum_{j=1}^{i} I(M_i; X_{j,u} \mid M_{j-1})\\
&\le \frac{1}{k} \sum_{i=1}^{n'}\sum_{j=1}^{i} I(M_i; X^{k}_{j}\mid M_{j-1})\\
&= \frac{\IC(\cA_k , \cD_{\text{int}})}{k}.
\end{align*}
Combining the above with Lemma~\ref{lem:single_information_lowerbound} provides a lower bound on $\IC(\cA' , \cD)$, which further implies the space lower bound
\[
\frac{1}{n'} \sum_{i=1}^{n'} \E[|M_i|] \ge \Omega(k\min\{\log(n/k) , \log\log(1/\delta)\})
\]
for the claimed range of $\eps$ and $\delta$.
\end{proof}

\subsection{Offline lower bound}
We now state and prove the remaining part of the lower bound.
\begin{lemma}
For any $\eps \in (0,c_1)$, $\delta \in (0, c_2)$ and $1<k\le n$,
any $(k,\eps, \delta)$-approximate counter $\cA$ that maintains $k$ counters with total sum at most $n$ must use
\[
|M| = \Omega\left(\min\{k\log(n/k), k\log(1/\eps) + k\log\log(n/k)\}\right)
\]
bits of memory.
\end{lemma}
\begin{proof}
Let $N(i) = \lceil (e^{4i \eps} -1)\cdot \eps^{-1}\rceil$.
Some simple calculations show
\begin{align*}
    (1-\eps)\cdot  N(i+1) - (1+\eps) \cdot N(i) &\ge \eps^{-1}(1-\eps)(e^{4 (i+1) \eps}-1) - \eps^{-1}(1+\eps)(e^{4 i \eps}-1) - (1+\eps)\\
    &= \eps^{-1}((1-\eps)e^{\eps} - 1-\eps)e^{4 i \eps} + 1 - \eps\\
    &\ge \eps^{-1}((1-\eps)e^{4\eps}-1-\eps)e^{4 i \eps} \\
    &\ge \eps^{-1}((1-\eps)(1+4\eps)-1-\eps)\cdot 1\\
    &= 2-4\eps > 0. 
\end{align*}
So for any $i \neq i'$, $N(i)$ and $N(i')$ must have a $(1\pm\eps)$ multiplicative gap.
We will consider $k$ counters that take on such values, i.e.\ $k$ counters that receive an (arbitrary) sequence of increments leading to counts $N({i_1}), N({i_2}), \dots, N({i_k})$ respectively. 
It is easy to see that if $i_r \le (4\eps)^{-1} \cdot \ln(1+n\eps/k)$ for all $r \in [k]$, the total sum of all the counters will be at most $n$.
Let $q := \lfloor (4\eps)^{-1} \cdot \ln(1+n\eps/k)  \rfloor$ and define the set of possible counts $\mathbf{N} := \{N(1), N(2), \dots, N(q)\}$.
We can represent the counts of all $k$ counters as vectors in $\mathbf{N}^k$. 

Consider a ``large'' collection of vectors $V \subseteq \mathbf{N}^k$ such that for every pair of counts $u, v \in V$, $u$ and $v$ differ in at least at least 90\% of the coordinates.
Notice that this implies that $u$ and $v$ differ multiplicatively in at least 90\% of the coordinates by definition.
Such a $V$ is equivalent to an error correcting code.
\begin{definition}
An error correcting code $C$ of length $k$ over a finite alphabet $\Sigma$ is a subset of $\Sigma^k$. The elements of $C$ are called \emph{code words}.
The distance of the code $C$, denoted $\Delta(C)$, is defined as the minimum hamming distance between any two code words $c_1, c_2 \in C$, i.e.
\[
\Delta(C) := \min_{\substack{c, c' \in C \\ c\neq c'} }\Delta(c, c'),
\]
where $\Delta(c, c') := |\{i : c_i \neq c_i' \}|$ is the hamming distance between two vectors.
\end{definition}
In the language of error correcting codes, we want our collection of counts $V$ to be a a large error correcting code with a minimum distance of $0.9k$.
Fortunately, the Gilbert-Varshamov bound immediately implies the existence of such a $V$ that is large enough for our purposes.
\begin{lemma}[Gilbert-Varshamov bound]
For any alphabet size $q > 1$, code length $k$ and distance $d \le k$, there exists an error correcting code $C$ with size,
\[
|C| \ge \frac{q^k}{\sum_{i=0}^{d-1}\binom{k}{i}(q-1)^i}.
\]
Consequently, for $d = 0.9 k$ and any $q$ larger than a universal constant, there is a code $C$ with size $|C| \ge q^{0.05k}$.
\end{lemma}
Thus, when $q$ is a large enough constant (which can be achieved by making for $\eps$ smaller than some universal constant $c$), the Gilbert-Varshamov bound tells us there is a choice of $V$ such that $|V| \ge q^{0.05k}$.
Fix a $(k, \eps, \delta)$-approximate counter $\cA$.
For any $v \in V$ and $i \in [k]$, $\cA$ must accurately approximate $v_i$ with probability at least $1-\delta$.
Denote the event that $\cA$ correctly approximates $v_i$ by $E_{v,i}$.
Since $\bE[\sum_{i=1}^k\mathbbm{1}(E^{c}_{v,i})] \le k \cdot \delta < k/20$, by Markov's inequality we can conclude the existence of a subset $I_v \subseteq [k]$ such that:
\begin{enumerate}
    \item $|I_v| \ge 0.9 k$,
    \item $\bP[\cap_{i \in I_v} E_{v,i} ] > 1/2$. Put in words, with probability at least $1/2$, for \emph{every} $i \in I_v$, the algorithm $\cA$ outputs a $(1 + \eps)$-approximation for $v_i$.
\end{enumerate}
Define the event $E_v := \cap_{i \in I_v} E_{v,i}$.
Since $\bE\left[ \sum_{v \in V}\mathbbm{1}(E^{c}_{v}) \right] \le |V|/2$, a standard averaging argument implies the existence of a fixed choice for the random seed of $\cA$ and a subset of counts $V' \subseteq V$ such that:
\begin{enumerate}
    \item $|V'| > |V|/2$,
    \item and the algorithm is correct on all $v \in V'$ in the sense of the event $E_{v}$ on this random seed.
\end{enumerate}
Fix such a random seed and denote the now deterministic algorithm $\cA'$.
For two counts $u, v \in V'$ define the set $D_{u,v} := \{i \in [k] : u_i \neq v_i\}$.
By construction we have that $|D_{u,v}| \ge 0.9k$.
We have
\begin{align*}
    |D_{u,v} \cap I_u \cap I_v| &\ge k - |D_{u,v}^c| - |I_u^c| - |I_v^c| \\
    &\ge k - 0.3k = 0.7k > 1.
\end{align*}
Thus, there is at least one index $i^* \in [k]$ such that $u_{i^*} \neq v_{i^*}$ \emph{and} $\cA'$ $(1+\eps)$-approximates both $u_{i^*}$ and $v_{i^*}$, so $\cA'$ arrives at a different state for $u$ and $v$.
Since this holds for every pair in $V'$, we can conclude that $\cA'$ must arrive at a different state for every count in $V'$.
We can now conclude that
\[
2^{|M|} \ge |V'| \ge 0.5\cdot(\lfloor (4\eps)^{-1} \cdot \ln(1+n\eps/k) \rfloor)^{0.05k} = \left(\Omega\left(\frac{\ln(1+n\eps/k)}{\eps}\right)\right)^{0.05k},
\]
or 
\begin{align*}
|M| \ge 0.05k\log\left(\frac{\ln(1+n\eps/k)}{\eps}\right) - O(k). 
\end{align*}
We distinguish into three cases: 
\begin{enumerate}
    \item If $\eps < k/n$, then 
\[
\frac{\ln(1+n\eps/k)}{\eps} = \Omega(n/k),
\]
which implies
\[
|M| \ge 0.05 k\log(n/k) - O(k).
\]
\item If $k/n \le \eps < \sqrt{k/n}$, we have 
\[
|M| \ge 0.05 k\log(1/\eps) - O(k) \ge 0.05 k\log(1/\eps) + k\log\log(n/k) - O(k\log\log(1/\eps)).
\]
\item If $\eps > \sqrt{k/n}$, we have
\[
|M| \ge 0.05  k\log(1/\eps) + 0.05  k\log\log(\eps n/k) - O(k) \ge 0.05k\log(1/\eps) + 0.05k\log\log(n/k) - O(k).
\]
\end{enumerate}

Putting all the above bounds together gives us
\begin{align*}
|M| &\ge \min\{0.05k\log (n/k) - O(k), 0.05k\log(1/\eps) + 0.025k\log\log(n/k) - O(k\log\log(1/\eps))\} \\
&= \Omega(\min\{k\log (n/k), k\log(1/\eps) + k\log\log(n/k)\}).
\end{align*}
\end{proof}

\section{Upper bounds}

We now state matching upper bounds $(k,\epsilon,\delta)$-approximate counting that follow immediately from the single counter case.
We give upper bounds in two regimes: the case $k \le N$, and the case $k > N$. We start by analyzing the first case.

Recall that the space usage of an approximate counter is typically a random variable, and the goal is to then prove an upper bound on the {\it expected} space (or, say, a high probability bound). The work of Nelson and Yu \cite{NelsonY22} provided the following bound on expected space usage for a single counter:

\begin{theorem}[{\cite{NelsonY22}}]\label{thm:ny22}
For any $\epsilon,\delta\in(0,1/2)$, there is an $(\eps,\delta)$-approximate counter with  expected space usage $O(\log\log N + \log\log(1/\delta) + \log(1/\epsilon))$ bits.
\end{theorem}

The following is then a very simple corollary of Theorem~\ref{thm:ny22}.

\begin{corollary}\label{cor:ub}
For any $\epsilon,\delta\in(0,1/2)$ and $1\le k\le N$, there is a $(k,\eps,\delta)$-approximate counter that uses $O(k(\log\log(2N/k) + \log\log(1/\delta) + \log(1/\epsilon)))$ bits in expectation. 
\end{corollary}
\begin{proof}
We simply instantiate $k$ independent copies of the data structure from Theorem~\ref{thm:ny22} to provide $k$ independent approximate counters, one per actual counter. For each $1\le i\le k$, let $S_i$ denote the (random) number of bits of memory used to store the approximate counter representing the $i$th counter $N_i$ and recall $N := \sum_i N_i$.
Then the total expected space is
\begin{align*}
    \E\left[\sum_{i=1}^k S_i\right] &= \sum_{i=1}^k \E [S_i]\\
    &\le C \sum_{i=1}^k (\log\log(N_i) + \log\log(1/\delta) + \log(1/\epsilon)) && \text{(Theorem~\ref{thm:ny22})}\\
    {}&= Ck(\log\log(1/\delta) + \log(1/\epsilon)) + \sum_{i=1}^k \log\log(N_i)\\
    {}&\le C k(\log\log(4N/k) + \log\log(1/\delta) + \log(1/\epsilon)) , && \text{(Jensen)} 
\end{align*}
where the last inequality uses concavity of the function $x\in (1,\infty) \mapsto \log\log(x)$. Note that we inject a constant $4$ in the iterated logarithm so that the $\log\log$ term stays nonnegative.
\end{proof}

We now turn to the case $k > N$. In this case, there is a clear lower bound of $\Omega(N\log(Ck/N))$ bits, since $\log(\sum_{j=1}^N \binom k j)$ bits of memory are needed to simply remember which counters are non-zero, and the logarithm of this sum is $\Theta(N\log(Ck/N))$ \cite[Exercise 9.42]{GrahamKP94}. We claim that this is also an upper bound. Specifically, we can use $O(N\log(Ck/N))$ bits to remember the locations of the $t\le N$ non-zero counters. Now we have reduced to the case $k=t \le N$ and can apply Corollary~\ref{cor:ub} to approximately remember their values.

\section*{Acknowledgements}
We thank Sidhanth Mohanty for very enlightening discussions on unpredictable paths, half Cauchy random variables, and stochastic processes in general that ultimately led to the discovery of the first version of our main lower bound.
We also thank Mark Sellke for answering a certain question regarding stochastic processes. Lastly, we thank Greg Valiant for raising the question of the amortized space complexity of approximate counting.


\begin{thebibliography}{BYJKS04}

\bibitem[BGW20]{BravermanGW20}
Mark Braverman, Sumegha Garg, and David~P. Woodruff.
\newblock The coin problem with applications to data streams.
\newblock In {\em Proceedings of the 61st {IEEE} Annual Symposium on
  Foundations of Computer Science (FOCS)}, pages 318--329, 2020.

\bibitem[BYJKS04]{BarYossefJKS04}
Ziv Bar-Yossef, T.~S. Jayram, Ravi Kumar, and D.~Sivakumar.
\newblock An information statistics approach to data stream and communication
  complexity.
\newblock {\em J. Comput. Syst. Sci.}, 68(4):702--732, 2004.

\bibitem[CSWY01]{ChakrabartiSWY01}
Amit Chakrabarti, Yaoyun Shi, Anthony Wirth, and Andrew~Chi{-}Chih Yao.
\newblock Informational complexity and the direct sum problem for simultaneous
  message complexity.
\newblock In {\em Proceedings of the 42nd Annual Symposium on Foundations of
  Computer Science (FOCS)}, pages 270--278, 2001.

\bibitem[CT06]{CoverT06}
Thomas~M. Cover and Joy~A. Thomas.
\newblock {\em Elements of information theory}.
\newblock Wiley, 2 edition, 2006.

\bibitem[GKP94]{GrahamKP94}
Ronald~L. Graham, Donald~E. Knuth, and Oren Patashnik.
\newblock {\em Concrete Mathematics: {A} Foundation for Computer Science, 2nd
  Ed}.
\newblock Addison-Wesley, 1994.

\bibitem[Kol18]{Kolevska18}
Elena Kolevska.
\newblock What happens with {Redis} runs out of memory, December 2018.
\newblock \url{https://www.youtube.com/watch?v=Xjq5XL2u3po}.

\bibitem[Lum18]{Lumbroso18}
J\'{e}r\'{e}mie Lumbroso.
\newblock The story of {HyperLogLog}: How {Flajolet} processed streams with
  coin flips.
\newblock {\em CoRR}, abs/1805.00612, 2018.

\bibitem[Mor78]{Morris78}
Robert~H. Morris.
\newblock Counting large numbers of events in small registers.
\newblock {\em Commun. {ACM}}, 21(10):840--842, 1978.

\bibitem[NY22]{NelsonY22}
Jelani Nelson and Huacheng Yu.
\newblock Optimal bounds for approximate counting.
\newblock In {\em Proceedings of the 41st ACM International Conference on
  Principles of Database Systems (PODS)}, pages 119--127, 2022.

\bibitem[Sip97]{Sipser97}
Michael Sipser.
\newblock {\em Introduction to the theory of computation}.
\newblock {PWS} Publishing Company, 1997.

\end{thebibliography}
\end{document}